\documentclass[twocolumn]{autart}
\usepackage[utf8]{inputenc}
\usepackage[T1]{fontenc}
\usepackage[english]{babel}
\usepackage{amsmath,amssymb,xcolor,booktabs}

\usepackage{enumitem}
\usepackage{refcount}
\usepackage[noadjust]{cite} 

\newcommand{\remove}[1]{ }
\newcommand{\ree}{\mathbb{R}}
\newcommand{\nat}{\mathbb{N}}
\newcommand{\LW}{{\cal L}_{\W}}
\newtheorem{sass}{Standing Assumption}

\newtheorem{theorem}{Theorem}[section]
\newtheorem{lemma}[theorem]{Lemma}
\newtheorem{remark}{Remark}[section]
\newtheorem{assumption}{Assumption}[section]
\newtheorem{example}{Example}[section]

\newtheorem{proposition}[theorem]{Proposition}

\newenvironment{proof}{\par\noindent\textit{Proof.}\ }{\hfill$\square$\par}

\usepackage{accents}

\let\leq\leqslant
\let\geq\geqslant

\let\tilde\widetilde
\let\cal\mathcal

\newcommand{\HS}{\cal H=(C,F,D,G,\W)}
\newcommand{\C}{C}
\newcommand{\D}{D}
\newcommand{\sol}{\phi}

\newcommand{\vcaexi}{VC-ae($\xi,w$) }
\newcommand{\vcexi}{VC-e($\xi,w$) }
\newcommand{\vcset}{VC-e($\Xi_0\cap \Pi_x(D^c,\W),\LW$) }
 
\newcommand{\R}{\ensuremath{\mathbb{R}}}
\newcommand{\B}{\ensuremath{\mathbb{B}}}
\newcommand{\N}{\ensuremath{\mathbb{N}}}

\newcommand{\W}{W}

\DeclareMathOperator*{\dom}{dom}
\DeclareMathOperator{\interior}{int}

\DeclareMathOperator{\graph}{graph}
\DeclareMathOperator{\range}{range}

\begin{document}
\begin{frontmatter}

\title{Solution Concepts and Existence Results for Hybrid Systems with Continuous-time Inputs\thanksref{footnoteinfo}}

\thanks[footnoteinfo]{This paper was not presented at any IFAC meeting. Research  partially funded by the European Research Council under the Advanced ERC Grant Agreement PROACTHIS, no. 101055384. Corresponding author: M. Heemels m.heemels@tue.nl.}

\author[Maurice]{W.P.M.H. Heemels}\ead{m.heemels@tue.nl},
\author[Romain]{R. Postoyan}\ead{romain.postoyan@univ-lorraine.fr},
\author[Pauline]{P. Bernard}\ead{ pauline.bernard@mines-paristech.fr},
\author[Koen]{K.J.A. Scheres}\ead{koen.scheres@kuleuven.be},
\author[Ricardo]{R.G. Sanfelice}

\address[Maurice]{Dept. Mechanical Engineering, Eindhoven University of Technology, The Netherlands} 
\address[Romain]{Universit\'e de Lorraine, CNRS, CRAN, F-54000 Nancy, France}
\address[Pauline]{Centre Automatique et Syst\`emes, Mines Paris - PSL, Paris, France} 
\address[Koen]{Department of Electrical Engineering (STADIUS), KU Leuven, Belgium}
\address[Ricardo]{Department of Electrical and Computer Engineering, University of California, Santa Cruz, CA 95064, USA}\ead{ricardo@ucsc.edu}


\begin{abstract}                          
In many scenarios, it is natural to model a plant's dynamical behavior using a hybrid dynamical system influenced by exogenous continuous-time inputs. While solution concepts and analytical tools for existence and completeness are well-established for autonomous hybrid systems, corresponding results for hybrid dynamical systems involving (continuous-time) inputs are generally lacking. This work aims to address this gap. We first  formalize notions of a solution for such systems. We then provide conditions to guarantee the existence and the forward completeness of solutions. To this end, we leverage viability theory results and present different conditions depending on the regularity of the exogenous input signals.
\end{abstract}
\end{frontmatter}

\section{Introduction}\label{sec:1}

Hybrid dynamical systems affected by external inputs arise in various scenarios such as hybrid control \cite{LooGru_AUT17a,nes_zac_tee_AUT08,pri_tar_zac_AUT13,prieur-teel-tac11,saez2025hybrid}, networked control \cite{HeeTee_TAC10a,NesTee_Aut_04,zhuang2024robust}, event-triggered control \cite{Scheres_Postoyan_Heemels_2024,BorDol_TAC18a,Mousavi_et-al-cdc2019} and (hybrid) state estimation \cite{petri-et-al-cdc2022,astolfi-et-al-tac2019(uniting),khalil2024hybrid} to name a few. Often, the input is given and is modeled as continuous-time (vector) signal representing measurement noise or external disturbances on the dynamics. It is therefore important to formalize what is meant by a solution to such systems and to determine conditions under which these solutions exist and are complete. For autonomous hybrid systems, these questions have been thoroughly addressed in the literature, see, e.g.,  \cite{Cortes,LygJoh:2003,HeeCam_BOOK_CONTRIB03a,Goe12,liberzon2003switching}. However, much remains to be done for hybrid systems with exogenous inputs.

In this context, we focus on hybrid dynamical inclusions \cite{Goe12} whose flow map, jump map, flow set and jump set can depend on exogenous inputs. Existing works on hybrid dynamical inclusion with inputs consider the case where either
\begin{enumerate}[label=(\roman*)]
    \item combinations of input and state trajectories (solution pairs) are sought and hence, the input is ``free'', see, e.g., \cite{Cai2009,Chai2019}; or,
    \item the input is ``given'' a priori and we search for a state trajectory only \cite{Bernard2020}.
\end{enumerate}
In all these cases, the inputs are given by so-called hybrid signals, meaning that they are defined on hybrid time domains involving both real time $t\in \ree_{\geq 0}$ and jump counter $j\in\nat$ (see Section \ref{subsect:class-systems-preliminaries}). Between two successive jumps, the input is Lebesgue measurable and locally essentially bounded. In \cite{Cai2009,Chai2019}, it is enforced that the domain of the input is the {\em same} as the domain of the solution itself. In certain cases this may be unnatural from a modeling perspective, as often the input is a priori given and the domain of the solution is a consequence of the hybrid dynamics and possibly the hybrid time domain of the input, but not {\em a priori} taken equal. In \cite{Chai2019}, a basic existence result of solution pairs is presented and the properties of maximal solutions (maximal in the sense that they cannot be prolonged) are characterized. Interestingly, as the domains of the input and the state have to be a priori the same, the existence result is, loosely speaking, stated as: ``given an initial state, there is a combination of an input signal and a state signal (on the same hybrid time domain) that form together a solution {\em pair}''. We are not aware of any result regarding the existence of solutions when the input signal is specified a priori for hybrid dynamical inclusions \cite{Goe12}.

Motivated by this gap in the literature, we study the scenario where the input to the system is {\em fixed} a priori and is {given by a continuous-time signal} defined on the nonnegative real line $\ree_{\geq 0}$. Indeed, for many applications, it would be welcome to have such existence and completeness results, as they are often needed to guarantee that the designed controllers (or estimators) are well-posed in presence of external disturbance and noise signals. We thus start by presenting two solution concepts: one where the solution belongs to the flow set at \textit{all} (continuous) times between two consecutive jumps (except at the beginning and end of the flow interval, to allow ``flowing'' from the boundary of the flow set), and another where it belongs only at \textit{almost all} times between two consecutive jumps. We then elaborate on the subtleties and consequences of using each possible variant of solution concepts. Afterwards, we provide tools to guarantee the existence of (nontrivial) solutions given an initial state and an input signal, and their (forward) completeness properties. Herein we differentiate between the variations in the solution concepts and the conditions imposed on the external input signal. Trajectory-dependent conditions are first presented for this purpose. To ease the testing of these conditions, we also provide trajectory-\emph{independent} conditions by exploiting viability results for non-autonomous differential inclusions \cite{Carjua2000viability,frankowska1996measurable,Aubin1984}. Different viability conditions are presented depending on the regularity of the exogenous input signal. As it appears that existing results do not allow the case where the flow condition depends on the input and the corresponding signal is only measurable, we provide a novel condition for this case which arises, for example, in event-triggered control \cite{Scheres_Postoyan_Heemels_2024}.

Compared to the preliminary version of this work \cite{heemels-cdc2021(hybrid-with-inputs)}, the main novel elements are the aforementioned trajectory-independent conditions to guarantee the existence of solutions. These novel results greatly simplify the test of the conditions for the existence of solutions. We also believe that the connection that is made with viability theoretical results for non-autonomous differential inclusions is relevant in its own right. In addition, we provide in this work all the proofs as well as new examples, which are in our view important to appreciate the subtleties of the presented results. Note that \cite{heemels-cdc2021(hybrid-with-inputs)} did not provide any proofs due to lack of space.

\noindent {\bf Notation:}
The sets of all nonnegative and positive integers are denoted $\N$ and $\N_{>0}$, respectively, and the set of rational numbers by $\mathbb{Q}$. The sets of reals and nonnegative reals are indicated by $\R$ and $\R_{\geq0}$, respectively. By $|\cdot|$ we denote the Euclidean norm. For $z\in \ree^n$ and $r\geq 0$, we denote by $\B(x,r)$ the closed ball in $\ree^n$ of radius $r$ around $z$, i.e., $\B(z,r) = \{x\in \ree^n \mid |x-z| \leq r \}$. For a set $A \subseteq \ree^n$, we denote the set of all points that lie within a distance of $r$ from $A$ by $\B(A,r)$, i.e., $\B(A,r) = \{x\in \ree^n \mid |x-a| \leq r \text{ for some } a \in A \}$. The interior of a set $A \subseteq \ree^n$ is denoted by $\interior{A}$ and its closure by $\overline{A}$. For sets $A$, $B \subseteq \ree^n$, we define the sum $A+B :=\{ a+b \mid a\in A\text{ and } b\in B\}$, the ``normal'' difference $A-B:=\{ a-b \mid a\in A\text{ and } b\in B\}$, and the Pontryagin difference $A\ominus B:= \left\{ x \in \ree^n \mid \{x\}+B \subseteq A\right\}$. The complement $S^c$ of a set $S \subseteq \ree^n$ is $S^c:=\{x\in \ree^n \mid x \not\in S\}$ and the (Bouligand) tangent cone $T_S(\xi)$ of $S$ at $\xi \in \ree^{n}$ is the set of all vectors $w\in \ree^{n}$ for which there exist $x_i\in S$, $\tau_i>0$ with $x_i\rightarrow \xi$, $\tau_i \rightarrow 0$ (when $i\rightarrow \infty$) and $w=\lim_{i \rightarrow} \frac{x_i-\xi}{\tau_i}$. For a set $M \subseteq \ree^{n_x} \times \ree^{n_w}$ and a set $\W\subseteq \ree^{n_w}$, we define the projection $\Pi_x(M,\W)$ as $\Pi_x(M,\W) := \{ \xi \in \ree^{n_x} \mid \exists w\in W \text{ s.t. } (\xi,w)\in M \}$.

\section{Solution concepts}
\subsection{Class of systems and preliminaries}\label{subsect:class-systems-preliminaries}

We study hybrid systems with inputs of the form
\begin{equation}
    \begin{cases}
        \dot{x}\in F(x,w) &(x,w)\in\C,\\
        x^+\in G(x,w) &(x,w)\in\D,
    \end{cases}\label{eq:hybridsys}
\end{equation}
where $x$ denotes the state taking values in $\R^{n_x}$, and $w$ the (disturbance) input taking values in $\W\subseteq \R^{n_w}$. Moreover, $\C\subseteq\R^{n_x}\times\W$ is the flow set, $\D\subseteq\R^{n_x}\times\W$ the jump set, $F:\R^{n_x}\times\W\rightrightarrows\R^{n_x}$ the flow map, and $G:\R^{n_x}\times\W\rightrightarrows\R^{n_x}$ the jump map, where $F$ and $G$ are possibly set-valued. Throughout this paper, we have the next standing assumption for technical reasons.

\begin{sass}[\textbf{SA\ref{sa:CW}}] \label{sa:CW} The sets {${\C}$ and ${\W}$ are closed. }
\end{sass}

To formally introduce the solution concept, we recall several notions from \cite{Cai2009,Chai2019,Goe12}. A subset $E\subset\R_{\geq0}\times\N$ is a {\em compact hybrid time domain}, if $E=\bigcup_{j=0}^{J-1}[t_j,t_{j+1}] \times \{j\}$ for some finite sequence of times $0=t_0\leq t_1\leq t_2\leq\ldots\leq t_J$. {It is a {\em hybrid time domain}, if for all $(T,J)\in E$, $E\cap([0,T]\times\{0,1,\ldots,J\})$ is a compact hybrid time domain}. A {\em hybrid signal} is a function defined on a hybrid time domain. For a hybrid time domain $E$, $\sup_t E:=\sup\left\{t\in\R_{\geq0}:\exists j\in\N\;\text{such that}\;(t,j)\in E\right\}$, $\sup_j E:=\sup\left\{j\in\N:\exists t\in\R_{\geq0}\;\text{such that}\;(t,j)\in E\right\}$ and $\sup E:=$ $(\sup_t E,\sup_j E)$. A hybrid signal $\sol$ is called a {\em hybrid arc}, if $\sol(\cdot, j)$ is locally absolutely continuous for each $j$ allowing to consider its continuous-time derivative $\dot\phi(\cdot,j)$ (almost everywhere) on $[t_j,t_{j+1}]$. Finally, the set of all Lebesgue measurable and locally essentially bounded functions from $\ree_{\geq 0} $ to $\R^{n_w}$ is denoted ${\cal L}$, where $n_w\in\nat_{>0}$. For $ \W\subseteq \R^{n_w}$ with $n_w \in \nat_{>0}$, ${\cal L}_{\W}$ is the set of functions $w : \ree_{\geq 0} \to {\W}$ that are Lebesgue measurable and locally essentially bounded.

\subsection{Definitions}\label{subsect:def}

Inspired by \cite{Cai2009,Chai2019}, we introduce the next two novel solution concepts.
\begin{defn} \label{eq:hybridsys_cai}
    A hybrid arc $\phi$ is an {\em e-solution} to $\HS$ for input $w\in\LW$ if
    \begin{description}
        \item[(S1-e)] for all $j\in\N$ such that $I^j:=\{t : (t,j) \in \dom \phi\}$ has nonempty interior, $\dot\phi(t,j)\in F(\phi(t,j),w(t))$ for almost all $t\in\interior I^j$ and $(\phi(t,j),w(t))\in\C$ for all $t\in\interior I^j$;
        \item[(S2)] for all $(t,j)\in\dom\phi$ such that $(t,j+1)\in\dom\phi$, $(\phi(t,j),w(t))\in\D$ and $\phi(t,j+1)\in G(\phi(t,j),w(t))$.
    \end{description}
    It is an {\em ae-solution} for input $w\in\LW$ if (S2) holds together with
    \begin{description}
        \item[(S1-ae)] for all $j\in\N$ such that $I^j$ has nonempty interior, $\dot\phi(t,j)\in F(\phi(t,j),w(t))$ and $(\phi(t,j),w(t))\in\C$ hold for almost all $t\in I^j$.
    \end{description}
\end{defn}

Some comments are in order. First of all, observe that an e-solution is also an ae-solution. Secondly, note that the ae-solution concept is in line with \cite[Section 2]{Cai2009}, where the flow constraint is satisfied {\em almost everywhere} during flow intervals -- though, as a difference to \cite[Section 2]{Cai2009}, an ae-solution is here defined for a given continuous-time input $w$. In contrast, in e-solutions, the flow constraint $(x,w)\in \C$ has to hold ``everywhere'' in the interior of flow intervals (but not necessarily at the boundaries of the flow interval). The e-solutions are closer in nature to \cite{Chai2019,Bernard2020} (although note that in \cite{Chai2019,Bernard2020,Cai2009,220} hybrid inputs on hybrid time domains are used instead of continuous-time functions defined on $\ree_{\geq 0}$). Thirdly, in \cite{Bernard2020}, (S1-e) is used in the notion of solution, but given the fact that during flow the input is required to be locally absolutely continuous, (S1-e) and (S1-ae) coincide in this setting. Moreover, generally we can state that if the input $w$ is continuous, then both solution concepts coincide as $\C$ is closed by SA\ref{sa:CW}. In fact, in the special case of no inputs, ae-solutions are e-solutions due to local absolute continuity of hybrid arcs during flow intervals, and we recover the solution concept in \cite{Goe12}. However, in a general setting where the input $w$ models a possibly discontinuous noise/disturbance, both definitions do not coincide as was illustrated by various examples in \cite{heemels-cdc2021(hybrid-with-inputs)}. One may choose to use one or the other depending on the context and the intended purposes. In any case, we are interested in this paper to study the existence and completeness properties for both e-solutions and ae-solutions given an input {in $\LW$}.

The following terminology will be used in this paper for both e-solutions and ae-solutions.

\begin{defn}
    An (e or ae-)solution $\phi$ to $\cal H$ for a given input $w\in{\LW}$ is called {\em nontrivial}, if $\dom\phi$ contains at least two points. It is said to be {\em maximal}, if there does not exist another solution $\psi$ to $\mathcal{H}$ for the same input $w$ such that $\dom \phi$ is a proper subset of $\dom\psi$ and $\phi(t,j)=\psi(t,j)$ for all $(t,j)\in\dom\phi$. We denote the set of all maximal e-solutions and ae-solutions to $\mathcal{H}$ for input $w$ by $\mathcal{S}^e_\mathcal{H}(w)$ and $\mathcal{S}^{ae}_\mathcal{H}(w)$, respectively. We say that the solution $\phi$ is {\em complete} if $\dom \phi$ is unbounded, and we say that it is {\em $t$-complete} if $\sup_t\dom \phi = \infty$.
\end{defn}

\section{Existence and completeness of solutions}
\label{sec:existence}

In this section, we start by providing conditions for the existence of nontrivial solutions and their properties for hybrid systems of the form \eqref{eq:hybridsys}, thereby extending the results in \cite[Proposition 2.10]{Goe12} without (external) inputs and in \cite[Proposition 3.4]{Chai2019} for a hybrid input to the case of continuous-time measurable inputs. This forms the main result of this section (Proposition \ref{prop:210}). After that, we provide examples illustrating the results and then present easier conditions to verify the hypothesis of the main result in this section.

The following ingredients are {essential in the following}. This concerns the following ``viability'' conditions VC-e($\Xi,\mathbb{W}$) and VC-ae($\Xi,\mathbb{W}$) for a set $\mathbb{W}\subset{\cal L_W} $ of input functions and a set $\Xi\subset \ree^{n_x}$ of initial states, given by

\begin{description}
   \item[VC-e($\Xi,\mathbb{W}$):] \hspace{0cm} For all $\xi \in \Xi$ and all $w \in \mathbb{W}$, there exist $\epsilon>0$ and an absolutely continuous function $z:[0,\epsilon]\rightarrow\R^{n_x}$ such that $z(0)=\xi$, $\dot z(t) \in F(z(t),w(t))$ for almost all $t\in[0,\epsilon]$ and $(z(t),w(t))\in {\C} $ for all $t\in(0,\epsilon)$.
   \item[VC-ae($\Xi,\mathbb{W}$):] \hspace{0cm} For all $\xi \in \Xi$ and all $w \in \mathbb{W}$ there exist $\epsilon>0$ and an absolutely continuous function $z:[0,\epsilon]\rightarrow\R^{n_x}$ such that $z(0)=\xi$, $\dot z(t) \in F(z(t),w(t))$ and $(z(t),w(t))\in {\C} $ for almost all $t\in[0,\epsilon]$.
\end{description}

In case $\Xi=\{\xi\}$ and $\mathbb{W}=\{w\}$ are singletons, we write VC-ae($\Xi,\mathbb{W}$) as VC-ae($\xi,w$), and similarly for VC-e.

\subsection{Existence of non-trivial solutions}

The next proposition provides necessary and sufficient conditions for the existence of non-trivial solutions.
\begin{proposition}\label{prop:210}
     Consider the hybrid system $\HS$.

    \noindent (i) {\bf [one initial state, one input]}\ There exists a nontrivial e-solution $\phi$ to $\cal H$ with input $w\in\LW$ and $\phi(0,0)=\xi\in\R^{n_x}$ if and only if $(\xi,w(0))\in\D$ or VC-e($\xi,w$) holds.

    \noindent (ii) {\bf [multiple initial states and inputs]}\ Let $\Xi_0\subseteq \R^{n_x}$ be given. For all $\xi\in\Xi_0 $ and all $w\in\LW$ there exists a nontrivial e-solution $\phi$ to $\cal H$ with input $w$ and $\phi(0,0)=\xi$ if and only if VC-e($\Xi_0\cap \Pi_x(D^c,\W),\LW$) holds.\\
    When VC-e {\em (i)} and {\em (ii)} is replaced by VC-ae, then all statements above apply to ae-solutions as well.
\end{proposition}
Note that in {\em (ii)} the direct coupling between $w(0)$ and $\xi$ as in {\em (i)} disappeared.

\begin{proof}
    Item {\em (i)} regarding the existence of a nontrivial solution follows directly from the definition of e-solutions to $\cal H$. Indeed, a solution can jump if and only if $(\xi,w(0))\in\D$ and a solution can flow on a nontrivial time window in the sense of e-solutions if and only if VC-e($\xi,w$) holds.

    To prove {\em (ii)}, we start by showing the ``only if'' part and assume the existence of  nontrivial e-solutions for all $\xi\in\Xi_0 $ and all $w\in\LW$. Due to {\em (i)} this implies $(\xi,w(0))\in\D$ or VC-e($\xi,w$) holds for all $\xi\in\Xi_0 $ and all $w\in\LW$. Take now $\xi\in \Xi_0\cap \Pi_x(D^c,\W)$ and $w\in \LW$. Due to the former, there exists an $\omega' \in \W$ such that $(\xi,\omega')\not\in\D$. Define now $\tilde w\in \LW$ by $\tilde{w}(0)=\omega' $ and $\tilde{w}(t)=w(t)$, $t>0$. Also for $\xi$ and $\tilde w$ a nontrivial e-solution exists. As $(\xi,\tilde{w}(0))\not\in\D$, VC-e($\xi,\tilde{w}$) has to hold, which is equivalent to VC-e($\xi,w$) (as the VC-e($\xi,w$) definition does not depend on $w$ at time $0$). Hence, VC-e($\Xi_0\cap \Pi_x(D^c,\W),\LW$) holds. The ``if'' part follows directly from {\em (i)}.
\end{proof}

\subsection{Properties of maximal solutions}

Let us now study the properties of maximal solutions. In formulating this result, the set
\begin{equation}
    {\C}_0 := \Pi_x(\C,\W)
\end{equation}
will play an important role. The set $C_0$ consists of all states $x$ in $\ree^{n_x}$ for which there exist a $w\in W$ such that $(x,w)\in C$.

\begin{lemma} \label{lem:1}
    Let $\xi\in\ree^{n_x}$ and $w\in\LW$ and suppose that \vcaexi holds. Then $z:[0,\epsilon)\rightarrow \ree^{n_x}$ coming from the satisfaction of \vcaexi satisfies $z(t)\in {\C}_0$ for all $t\in [0,\epsilon]$.
\end{lemma}
\begin{proof}
    Take $t\in [0,\epsilon]$. Since $z$ is absolutely continuous and $w$ is essentially bounded (and thus $w$ takes values in a compact set for almost all times), combined with the fact that $(z(t),w(t))\in {\C} $ for almost all $t\in[0,\epsilon]$, there is a sequence $\{t_l\}_{l\in\nat}$ with $t_l\in (0,\epsilon)$ for each $l\in\nat$, and $t_l \rightarrow t$ when $l \rightarrow \infty$ satisfying that $w(t_l) \rightarrow \bar w$ for some $\bar w$ and $(z(t_l),w(t_l))\in {\C} $ for each $l\in\nat$. Clearly, since $w(t_l)\in \W$ for each $l\in\nat$, and $\W$ is closed, we have $\bar w\in {\W}$. Since $z$ is continuous, $z(t_l) \rightarrow z(t)$ when $l \rightarrow \infty$, and since ${\C}$ is closed, we have that $(z(t),\bar{w}) \in {\C}$. This establishes that $z(t)\in \C_0$.
\end{proof}
Some interesting remarks are in order with respect to ${\C}_0$. First, from Lemma~\ref{lem:1} it follows that satisfaction of VC-ae$(\xi,w$) or VC-e$(\xi,w$) for some $w \in \LW$ implies $\xi\in{\C}_0$. Hence, for $w \in \LW$ and $\xi\not\in {\C}_0 $ such that $(\xi,w(0))\not\in \D$, the set of (nontrivial) solutions $\sol$ with $\sol(0,0)=\xi$ and input $w$ is empty. Second, the lemma shows that during flow the state will not leave ${\C}_0$. So, even a discontinuity in $w$ at $t=\epsilon$ cannot get the state outside ${\C}_0$.

The next notions will play an important role when analyzing the domains of maximal solutions.
\begin{defn}
    For $a\in\R_{\geq0}$, define the left shift operator $S_a:{\cal L}\rightarrow {\cal L}$ for $w\in{\cal L} $ as $S_a(w) := \tilde w$ , where $\tilde w(t) = w(t+a)$ for all $t\geq 0$. Moreover, for $\mathbb{W}\subset{\cal L_W} $ and $a\in\R_{\geq0}$, we define $S_a(\mathbb{W}):=\{S_a(w)\mid w\in\mathbb{W}\}$ and $S(\mathbb{W}):=\bigcup_{a\in\R_{\geq0}}S_a(\mathbb{W})$.
\end{defn}
Loosely speaking, the set $S(\mathbb{W})$ is the collection of all left shifted versions of functions in $\mathbb{W}$. The next proposition establishes properties of the maximal solutions to hybrid system (\ref{eq:hybridsys}). It can be seen as an extension of \cite[Proposition 2.10]{Goe12} to hybrid systems with continuous-time inputs.

\begin{proposition}\label{prop:210b}
    Consider the hybrid system $\HS$.

    \noindent (i) {\bf [one input]}\ Let input $w\in\LW$ be given. If for each $a \geq 0$ and each $\xi\in{\C}_0$ it holds that $(\xi,w(a))\in\D$ or VC-e($\xi,S_a(w)$), then every maximal e-solution $\phi\in\mathcal{S}_\mathcal{H}^e({w})$ satisfies exactly one of the following properties:
    \begin{enumerate}[label=(\alph*)]
        \item $\phi$ is complete;
        \item $\phi$ is not complete and ``ends with flow'': $\dom \phi$ is bounded and the interval $I^J:=\{t:(t,J)\in\dom\phi\}$ with $J=\sup_j\dom\phi$ is open to the right, and there does not exist an absolutely continuous function $z:\overline{I^J}\rightarrow\R^{n_x}$ satisfying $\dot z(t)\in F(z(t), w(t))$ for almost all $t\in I^J$ and $(z(t), w(t))\in\C$ for all $t\in \interior I^J$, and such that $z(t)=\phi(t,J)$ for all $t\in I^J$;
        \item $\phi$ is not complete and ``ends with a jump'': $\dom\phi$ is bounded with $(T,J):=\sup\dom\phi\in \dom \phi$, $I^J=\{T\}$, $(\phi(T,J), w(T))\not\in\D$ and $\phi(T,J)\not\in{\C}_0$.
    \end{enumerate}

    \noindent (ii) {\bf [multiple inputs]}\ If VC-e(${\C}_0\cap \Pi_x(\D^c,\W),\LW$) holds, then all maximal solutions $\phi$ to ${\cal H}$ satisfy (a), (b), or (c) above.

    When VC-e is replaced by VC-ae, then all the above statements apply to ae-solutions, where in (b) the phrase ``$(z(t), w(t))\in\C$ for all $t\in I^J$'' is replaced by ``$(z(t), w(t))\in\C$ for almost all $t\in I^J$.''
\end{proposition}

\begin{proof}
    To prove \emph{(i)}, suppose that $\phi\in \mathcal{S}^e_\mathcal{H}(w)$ is a maximal solution that is not complete, i.e., $\dom\phi$ is bounded. Let $(T,J)=\sup\dom\phi$. If $(T,J)\in\dom\phi$ and $(\phi(T,J),w(T))\in\D$ or $\phi(T,J)\in {\C}_0$, then either $(\phi(T,J),w(T))\in\D$ in which case $\phi$ can be extended via a jump, or $\phi(T,J)\in {\C}_0$ in which case $\phi$ can be extended via flow. The latter follows thanks to VC-e$(\xi,S_a(w))$ holding for $a=T$ and $\xi=\phi(T,J) \in {\C}_0$. This would contradict maximality of $\phi$. Therefore, if $(T,J)\in\dom\phi$, necessarily $(\phi(T,J),w(T))\notin\D$ and $\phi(T,J)\notin {\C}_0$. But according to Lemma~\ref{lem:1}, we cannot have $\phi(T,J)\notin{\C}_0$ due to flow (even if there is a jump in $w$) so necessarily $I^J=\{T\}$ and the solution ends with a jump. Thus, either (c) holds or $(T,J)\not\in\dom\phi$. If the latter holds, then $I^J$ has to be right open, and the e-solution necessarily ends with flow since a jump necessarily leads to $(T,J)\in\dom\phi$. Besides, (b) must hold to ensure maximality of $\phi$. Indeed, if (b) would fail, the e-solution $\phi$ could be extended to an e-solution to $\cal H$ on $\overline{\dom\phi}$. Similar arguments lead to proving (ii) noting that $S(\LW)=\LW$. The ``ae-case'' can be proven similarly.
\end{proof}

\subsection{Examples}

Let us illustrate the use of Propositions~\ref{prop:210} and \ref{prop:210b}.
\begin{example} \label{ex:noise}
     Consider the {\em event-driven control system} in \cite{Astrom1999comparison} in which we removed the process noise and introduce measurement noise that takes values in $\W:= [-0.2,0.2]$, giving
    \begin{equation} \label{ex:h2}
       \begin{cases}
           \dot x\ \ = x& | x + w| \leq 1.5,\\
           x^+ = -w & | x + w| \geq 1.
        \end{cases}
    \end{equation}
    The resets of the state in \cite{Astrom1999comparison} lead to $x^+=0$, i.e. resets to zero due to an impulsive Dirac spike of intensity $x$, which was possible as the {\em exact} state was known. However, here we only know a noisy measurement $x+w$, so in a reset we can only apply an impulsive Dirac spike with intensity $-(x+w)$ resetting $x$ prior to the reset to $-w$ after the reset. Overall, the objective of this event-driven reset based on noisy measurements is aimed to keep the state inside the set $[-1.7,1.7]$. To apply Proposition~\ref{prop:210}, note that ${\C}_0= [-1.7,1.7]$. For $\xi\in{\C}_0$ and $w\in\LW$ with $(\xi,w(0))\not\in \D$, we have $|\xi|<1.2$ since $|\xi+w(0)|<1$ and $|w(0)|\leq 0.2$. Given the definition of the flow set, we see that the flow condition holds as long as $|x|\leq 1.3$. This shows that \vcexi and \vcaexi are both satisfied as it takes some positive time for the solution to $\dot x=x$ to grow from an absolute value below $1.2$ to an absolute value larger or equal than $1.3$. Hence, the existence of nontrivial ae- and e-solutions is guaranteed and any maximal solution satisfies only one of the properties (a)-(c) in Proposition~\ref{prop:210b}. In fact, it can be checked that only (a) can hold, since each maximal interval of flow $I^j$ is necessarily closed, and $|x+w|\leq 0.4$ after a jump. Therefore, all maximal (e- or ae-) solutions to $\cal H$ for input $w$ are complete -- in fact, $t$-complete. This property holds for any $w\in \LW$. Note that for any $w\in\mathcal{L}_{\W}$, each (e- and ae-) solution $\phi$ remains indeed in $[-1.7,1.7]$, so forward invariance of this set is guaranteed despite the presence of measurement noise.
\end{example}

However, there are some subtleties with Proposition~\ref{prop:210} due to the use of merely measurable inputs, as will be demonstrated with the next reworked example, where the overlap between ${\C}$ and ${\D}$ is smaller than in Example~\ref{ex:noise}.

\begin{example} \label{ex:noise2}
    We embed Example~\ref{ex:noise} (with $c=1.5$) in
    \begin{equation} \label{ex:h3}
        \begin{cases}
            \dot x\ \ = x& | x + w| \leq c,\\
            x^+=-w& | x + w| \geq 1.
        \end{cases}
    \end{equation}
    but now set $c=1$ and still adopt $ \W:= [-0.2,0.2]$. Even though the union of the jump and flow sets is the full space and both are closed, maybe counterintuitively, it is not the case that nontrivial (e- or ae-)solutions exist for all initial states and Lebesgue measurable input functions. For example, for $\sol(0,0)=\xi=1$ and $w\in \LW$ given by
    \[
        w(t) =  -0.2 \text{ for } t=0, \text{ and } w(t) = 0.2 \text{ for } t>0,
    \]
    neither jump nor flow are possible. Indeed, a jump is not possible as the jump condition $| \sol(0,0) + w(0)| \geq 1$ is not satisfied and neither \vcexi nor \vcaexi is satisfied, as, loosely speaking, $| x + w| \leq 1$ cannot be satisfied on a nontrivial flow interval. Essentially, a problematic issue for the existence of nontrivial solutions is that the possibility of a jump at $t = 0$ depends on $w(0)$ (i.e., $(\xi,w(0))\in \D$ is needed) while the value of $w(0)$ does not affect the satisfaction of \vcexi or \vcaexi. Informally, the information $(\xi,w(0))\not\in \D$ does not restrict $\xi$ sufficiently (as in Example~\ref{ex:noise}) to conclude that flow is possible for $w(t), \ t>0$. Similar problems arise for \eqref{ex:h3} as long as $c< 1.4$. Indeed, take $c= 1.4 - \eta$ for small positive $\eta$ and consider $\sol(0,0)=1.2-\frac{\eta}{2}$ and $w$ as above. No flow is possible (neither \vcexi nor \vcaexi holds) and a jump is not possible either.
\end{example}

In fact, if there are no restrictions on the set of inputs $\mathbb{W}\subset{\cal L}$, i.e., ${\W}=\ree$ and thus $\mathbb{W}={\cal L}$, then ${\C}_0=\ree$ in the example \eqref{ex:h3} above, and \vcexi or \vcaexi has to hold for all $w\in \mathcal{L}$ and all $\xi\in \Pi_x(D^c,\mathbb{R})=\ree$ in order to have existence of nontrivial solutions, which is not satisfied for this example and typically also not for many other systems. So, a tighter choice of ${\mathbb{W}}$ is needed, either through
\begin{itemize}
    \item using a sufficiently small set $\W$ in combination with overlap between the flow and jump map, e.g., for system \eqref{ex:h3} with $c>1$ and $W=[-\delta,\delta]$ with $0 \leq \delta \leq \frac{c-1}{2}$, as seen also in Example~\ref{ex:noise}, or
    \item imposing more regularity on the functions, as the use of the rich class of Lebesgue measurable inputs for which the value of $w(0)$ (important for jumps) is not related in any way to the function values $w(t)$ for times $t>0$ (important for flow, i.e., for the satisfaction of VC-ae or VC-e), renders existence of nontrivial solutions difficult. In fact, by restricting the input signals in $\mathbb{W}$ to be piecewise continuous, we can guarantee the existence of nontrivial (e- and ae-) solutions for all systems \eqref{ex:h3} with $c\geq 1$ for any $\W$. We will discuss this in Section~\ref{sec:pc}.
\end{itemize}

\section{Restricting the inputs to c\`adl\`ag signals} \label{sec:pc}

As we have seen in Example \ref{ex:noise2}, the existence of nontrivial solutions for all Lebesgue measurable inputs is often hard to guarantee in scenarios where $\C$ directly depends on $w$, apart from rather particular cases. One such particular case is Example \ref{ex:noise}, where the bounds on $w$ and the sufficiently large overlap between $\C$ and $\D$ allow us to guarantee the existence of nontrivial solutions for all $w\in {\cal L}_{\W}$. This particular case will be discussed in more detail in Section~\ref{subsubsec:measu}. In this section, we will show that under additional regularity assumptions on the inputs, namely piecewise continuity, or to be precise c\`adl\`ag, which still forms a rather rich class of inputs, the existence of nontrivial solutions can be more easily guaranteed, including systems such as \eqref{ex:h3}.

\begin{defn} \label{def:pc}
    A function $w : \R_{\geq0} \rightarrow \R^{n_w}$ is said to be {\em c\`adl\`ag} (``continue \`a droite, limite \`a gauche''), denoted by $w \in\cal{PC}$, when there exists a sequence $\{t_i\}_{i\in\N}$ with $t_{i+1}>t_i >t_0 =0$ for all $i\in\N$ and $t_i\rightarrow\infty$ when $i\rightarrow\infty$ such that $w$ is continuous on $(t_i,t_{i+1})$, where $\lim_{t\uparrow t_i} w(t)$ exists for all $i\in\N_{>0}$ and $\lim_{t\downarrow t_i} w(t)$ exists for all $i\in\N$ with $\lim_{t\downarrow t_i} w(t) = w(t_i)$, i.e., $w$ is piecewise continuous, right continuous and left limits exist for each $t_i$, $i\in\N_{>0}$. Given a set ${\W}\subseteq \R^{n_w}$, then we denote by $\cal{PC}_{\W}$ the set of functions $\{w\in \cal{PC} \mid w(t) \in {\W} \text{ for all } t \in \ree_{\geq 0}\}$.
\end{defn}
Note that continuous functions are contained in $\mathcal{PC}$ since $\{t_i\}_{i\in\mathbb{N}}$ can then be chosen arbitrarily. In the next proposition we provide conditions for the existence of nontrivial solutions for hybrid system \eqref{eq:hybridsys}, with the $\cal{PC}$ restriction on inputs. Interestingly, checking VC-ae or VC-e for a $w\in \cal{PC}$, is equivalent to checking it for a {\em continuous} input as the restriction of $w$ to a small enough interval is
continuous. Hence, since we consider closed sets $\C$, this causes VC-e and VC-ae to coincide, and, therefore, we will have only one viability condition defined as
\begin{description}
    \item[VC$(\xi,w)$:] \hspace{0.5cm} There exist $\epsilon>0$ and an absolutely continuous function $z:[0,\epsilon]\rightarrow\R^{n_x}$ such that $z(0)=\xi$, $\dot z(t) \in F(z(t),w(t))$ for almost all $t \in [0,\epsilon]$ and $(z(t),w(t))\in {\C}$ for all $t\in[0,\epsilon]$.
\end{description}
Similarly, we can show that in this context, e- and ae-solutions actually coincide.

\begin{proposition}\label{prop:210_v2}
    Consider the hybrid system $\HS$.
    \begin{enumerate}[label=(\roman*)]
        \item Any ae-solution to $\cal H$ with input $w\in \cal{PC}_{\W}$ is also an e-solution, hereafter called ``solution.''
        \item There exists a nontrivial solution $\phi$ to $\cal H$ for input $w\in\cal{PC}_{\W}$ with $\phi(0,0)=\xi\in\R^{n_x}$ if and only if $(\xi,w(0))\in\D$ or VC$(\xi, w)$ holds.
        \item If condition VC$(\xi,\bar w)$ holds for all $\xi\in\R^{n_x}$ and all $\bar w\in\mathcal{ PC}_{\W}$ with $(\xi,\bar w(0))\in{\C}\setminus\D$, then for all $ w\in\mathcal{ PC}_{\W}$ every maximal solution $\phi\in\mathcal{S}_\mathcal{H}^e(w)$ satisfies exactly one of the following properties:
    \end{enumerate}
    \begin{enumerate}[label=(\alph*)]
        \item $\phi$ is complete;
        \item $\phi$ is not complete and ``ends with flow'': $\dom \phi$ is bounded and the interval $I^J:=\{t:(t,J)\in\dom\phi\}$ with $J=\sup_j\dom\phi$ is open to the right, and there does not exist an absolutely continuous function $z:\overline{I^J}\rightarrow\R^{n_x}$ satisfying $\dot z(t)\in F(z(t), w(t))$ for almost all $t\in I^J$ and $(z(t),w(t))\in\C$ for all $t\in \interior I^J$, and such that $z(t)=\phi(t,J)$ for all $t\in I^J$;
        \item $\phi$ is not complete and ``ends with a jump'' or a ``discontinuity'' of $w$: $\dom\phi$ is bounded with $(T,J):=\sup\dom\phi\in \dom \phi$, $(\phi(T,J), w(T))\not\in{\C}\cup {\D}$.
    \end{enumerate}
\end{proposition}

\begin{proof}
    To prove \emph{(i)}, consider an ae-solution $\phi$ to $\cal H$ with $w\in \cal{PC}_{\W}$. This ae-solution is not an e-solution, if there exists a $j\in \dom \phi$ with $t \in \interior I^j$ such that $(\phi(t,j),w(t))\notin \C$. However, then, by closedness of $\C$, $(\phi(t,j),w(t))$ lies at a positive distance of ${\C}$. Therefore, by right-continuity of $w$ and continuity of $\phi(\cdot,j)$, $(\phi(t,j),w(t))$ cannot be contained ${\C}$ almost everywhere in $I^j$, contradicting that $\phi$ is an ae-solution. Hence, $\phi$ has to be an e-solution.

    Statement \emph{(ii)} regarding the existence of a nontrivial solution follows directly from the definition of a solution to $\cal H$ as solutions can either jump or flow, see also the proof of Proposition~\ref{prop:210}. Regarding the properties of maximal solutions in \emph{(iii)}, suppose that $\phi$ is a maximal solution that is not complete, i.e., $\dom\phi$ is bounded. Let $(T,J)=\sup\dom\phi$. If $(T,J)\in\dom\phi$ and $(\phi(T,J),w(T))\in{\C}\cup\D$, then either $(\phi(T,J),w(T))\in\D$ in which case $\phi$ can be extended via a jump, or $(\phi(T,J),w(T))\in{\C}\setminus\D$ in which case $\phi$ can be extended via flow, thanks to VC$(\xi,\bar w)$ holding for all $\bar w\in \mathcal{ PC}_{\W}$ with $(\xi,\bar w(0))\in{\C}\setminus\D$. Therefore, if $(T,J)\in\dom\phi$, necessarily $(\phi(T,J),w(T))\notin\C\cup \D$. According to the reasoning above, this can only happen after a jump or at the end of an interval of flow where $w$ is discontinuous. Thus, either (c) holds or $(T,J)\not\in\dom\phi$. If the latter holds, then the interior of $I^J$ is nonempty, since we could not get to $\phi(T,J)$ via a jump which would cause $(T,J)\in\dom\phi$, and (b) must hold to ensure maximality of $\phi$. Indeed, if (b) would fail, the solution $\phi$ could be extended to a solution to $\cal H$ on $\overline{\dom\phi}$.
\end{proof}

As discussed before Proposition \ref{prop:210_v2}, note that we only have to verify in
VC$(\xi,w)$ a property over an interval $[0,\epsilon]$ of continuity of $w$, and, hence, the VC conditions in \emph{(ii)} and \emph{(iii)} have only to be guaranteed for {\em continuous} inputs.

To illustrate the rationale behind our choice for c\`adl\`ag inputs, let us reconsider Example~\ref{ex:noise2} demonstrating that imposing $\mathcal{PC}$-regularity on the inputs leads to existence of nontrivial solutions as well as completeness of maximal solutions, while this was not the case for measurable inputs.

\begin{example} \label{ex:w0_revisit}
    We revisit Example~\ref{ex:noise2} and, in particular, system \eqref{ex:h3} with $c=1$, where the only difference will be that the input signals are in $\mathcal{PC}_{\W}$ instead of in $\LW$. With this restriction, we show that the existence of nontrivial solutions is now guaranteed for all initial states and all inputs in $\mathcal{PC}_{\W}$ based on Proposition~\ref{prop:210_v2}. Note that this statement did not hold when we considered inputs in $\LW$, as demonstrated in Example~\ref{ex:noise2}. If we compute ${\C}\setminus\D$ we obtain $\C\setminus\D= \{(\xi,w) \in \ree^2 \mid |\xi + w| <1\}$. As for any $w\in \mathcal{ PC}_{\W}$ there is an $\epsilon>0$ such that $w$ is continuous on $[0,\epsilon]$ and the solution to $\dot x=x$ is (locally absolutely) continuous too, we see that for every $\xi\in\R^{n_x}$ and $w\in \mathcal{PC}_{\W}$ with $(\xi,w(0))\in{\C}\setminus\D$, VC$(\xi,w)$ holds due to the strict inequality in the expression for $\C\setminus\D$. It takes some positive time for the continuous function $t\mapsto (z(t),w(t))$ to leave ${\C}$. Hence, the existence of nontrivial (ae- and e-)solutions is guaranteed (in fact, for $\W=\ree$) and, moreover, in this case we can even show that each maximal solution is complete, and even $t$-complete. This clearly contrasts the case of Lebesgue measurable inputs in which these existence and completeness properties did not hold for the same hybrid system.
\end{example}

\begin{remark} \label{rem:2}\label{Ex:disc}
Extra care is required regarding item (c) of Proposition~\ref{prop:210_v2}, if compared to the non-input case. In hybrid systems without inputs, it would be sufficient to prove that $G(\D)\subset \C\cup\D$ to exclude (c). However, when dealing with discontinuous inputs as defined above, discontinuities in $w$ could also result in (c) occurring.
To illustrate this phenomenon, consider system
\begin{equation} \label{ex:h4}
       \begin{cases}
           \dot{x}\ \ =-x-w& x + w \leq 1,\\
           x^+=-w & -2 \leq x+ w \leq 2.
        \end{cases}
    \end{equation} Take $\W = \ree$. It can be shown using Proposition~\ref{prop:210_v2} that for every initial state $\xi$ and every $w\in \mathcal{PC}_{\W}$ with $(\xi,w(0))\in {\C}\cup \D$ a nontrivial solution exists. However, not all maximal solutions are complete due to case (c) mentioned in Proposition~\ref{prop:210_v2} occurring as a result of discontinuities in $w$. Indeed, take $\sol(0,0)=\xi=1$ and $w(t) = -1$ if $t \in [0,1)$ and $w(t) = 2$ if $t \geq 1$. For this choice, a maximal solution is given by $\sol(t,0) = 1, \ t \in [0,1]$ on the hybrid time domain $[0,1]\times \{0\}$. The discontinuity in $w$ at time $1$ leads to $(\sol(1,0),w(1))\not\in{\C}\cup\D$. Note that jumps according to the jump map {\em cannot} lead to case (c) for this example.
\end{remark}

\section{Trajectory-independent viability conditions}
\label{sec:trajectind}

The viability conditions considered so far, namely VC-e or VC-ae, may not be easy to test as they are trajectory-dependent in the sense that they involve a solution $z$ to the differential inclusion $\dot z \in F(z,w)$. Hence, more explicit trajectory-independent conditions involving the tangent cone to $C$, that can be used to guarantee the viability properties of VC-e or VC-ae, are provided in this section. We start with inputs $w$ that are (piecewise) continuous and connect to Section~\ref{sec:pc} above. After that, we consider the more difficult case of measurable inputs.

We rely on the next extra properties throughout this section, where we use the definitions of outer semicontinuity and local boundedness (relative to a set) of $F$ as in \cite[Definitions 5.9 and 5.14]{Goe12}, respectively.

\begin{assumption} \label{sa:reg}
    Given $\HS$. The mapping $F:\R^{n_x}\times\R^{n_w}\rightrightarrows\R^{n_x}$ satisfies the following properties.
    \begin{enumerate}[label=(\roman*)]
        \item $F:\R^{n_x}\times\R^{n_w}\rightrightarrows\R^{n_x}$ is outer semicontinuous (osc) and locally bounded relative to $\C$.
        \item $F(\zeta,\omega)$ is non-empty and convex for each $(\zeta,\omega)\in\C$ with $\omega\in\W$.
    \end{enumerate}
\end{assumption}

Notice that Assumption \ref{sa:reg}(i) implies that $F(\zeta,\omega)$ is compact for each $(\zeta,\omega)\in\C$ with $\omega\in\W$, as closedness follows from SA1.
\subsection{Absolutely continuous and continuous inputs}

In this section, we exploit continuity properties of the continuous-time inputs to obtain explicit tangent cone conditions guaranteeing the viability condition VC-e($\xi,w$) for a given initial state $\xi$ and input $w$. As $w$ is assumed to be continuous, we have that the notions of ae- and e-solutions coincide as mentioned in Section \ref{subsect:def}.

\subsubsection{Absolutely continuous inputs}
The next proposition provides the desired tangent cone-based condition to ensure VC-e.

\begin{proposition} \label{prop:v:ac}
    Consider $\HS$ with Assumption~\ref{sa:reg} holding. Let $\xi\in\R^{n_x}$ and $w\in\LW$ be absolutely continuous (AC) with its derivative essentially bounded on an interval $[0,\varepsilon]$ for some $\varepsilon>0$. If there exists a neighborhood $U$ of $(\xi,w(0))$ such that for all $(\zeta, \omega) \in \C \cap U$ and almost all ${\tau}\in[0,\varepsilon)$
    \begin{equation}
        (F(\zeta, \omega) \times \{\dot w(\tau)\}) \cap T_C(\zeta, \omega) \neq \emptyset, \label{eq:vc-tc-AC}
    \end{equation}
    then VC-e($\xi,w$) holds.
\end{proposition}

\begin{proof}
    The proof relies on \cite[Thm.~2.3]{Carjua2000viability} by embedding $t\mapsto w(t)$ as a state in an appropriate constrained differential inclusion. Thereto, denote the derivative of $w$ as $v$, such that $\dot{w}(t)=v(t)$ for almost all $t\in[0,\varepsilon]$ and let $U\subset\R^{n_x}\times W$. We define a new state variable $\chi=(z,w)$ and embed $\dot z\in F(z,w(t))$ with $(z,w(t))\in C\cap \overline U$ in
    \begin{equation} \label{eq:Hr}
        \dot\chi \in H_r(t,\chi) := \left\{\begin{array}{lllll}
            H(t,\chi) & \chi\in C\cap \interior(U) \\
            \operatorname{con}\left(H(\chi,t)\cup\{0\}\right) & \chi\in C\cap \partial U
        \end{array}\right.
    \end{equation}
    with $H(t,\chi):=F(z,w) \times \{v(t)\}$, $\operatorname{con}$ denoting the convex hull, and $\partial U:=\overline U \setminus \interior(U)$ the boundary of $U$. Note that $H$ satisfies (A)-(D) in \cite{Carjua2000viability} taking the set $D$ in \cite[Thm.~2.3]{Carjua2000viability} as $D=\C \cap \overline{U}$. Indeed, the osc and local boundedness of $F$ (due to Assumption~\ref{sa:reg}) guarantee the upper semicontinuity of $H(t,\cdot)$ in $C$ (condition (B)) for almost all $t$ using \cite[Lemma~5.15]{Goe12}. Moreover, for every $\chi \in C \cap \overline{U}$ the set-valued map $H(\cdot,\chi)$ is measurable on $[0,\varepsilon)$ (condition (A) in \cite{Carjua2000viability}) because $v$ is measurable as $w$ is AC. Conditions (C) and (D) in \cite{Carjua2000viability} follow from Assumption~\ref{sa:reg}.

    Conditions (A)-(D) in \cite{Carjua2000viability} also hold for $H_r$, as $H_r(t,\chi)=\operatorname{con}(H(t,\chi)\cup L(\chi))$ with $L$ defined as $L(\chi)=\emptyset$ for $\chi\in C\cap\interior(U)$ and $L(\chi)=\{0\}$ for $\chi\in C\cap\partial U$. Indeed, the set-valued map $H(t,\chi)\cup L(\chi)$ is measurable in time $t$ for any $\chi\in C\cap\overline U$ by \cite[Proposition 14.11(b)]{rockafellar-wets-book} as so are $H(t,\chi)$ and $L(\chi)$. As a consequence, $\operatorname{con}(H(t,\chi)\cup L(\chi))$ is measurable in $t$ for all $\chi\in C\cap\overline U$ by \cite[Exercise 14.12]{rockafellar-wets-book}. Hence, condition (A) in \cite{Carjua2000} holds for $H_r$. On the other hand, for almost all $t\in[0,\varepsilon)$, $H(t,\cdot)\cup L(\cdot)$ is locally bounded (as its maps bounded sets to bounded sets \cite[Proposition 5.15]{rockafellar-wets-book}) and is outer semicontinuous being the pointwise union of two outer semicontinuous set-valued maps. Consequently, $H_r(t,\cdot)$ is locally bounded and outer semicontinuous for almost all $t\in[0,\varepsilon)$ \cite[Exercise 4.18]{goebel-book2024}, which proves that condition (B) holds for $H_r$. Conditions (C) and (D) then follow. The last condition to check is (V1) in \cite[Thm.~2.3]{Carjua2000viability}. Property (\ref{eq:vc-tc-AC}) implies that for any $\chi\in C\cap\interior(U)$ and almost all $t\in[0,\varepsilon)$, $H_r(t,\chi)\cap T_{C\cap\overline U}(\chi)\neq \emptyset$ as in this case $H_r(t,\chi)=H(t,\chi)=F(\zeta,w)\times\{\dot w(t)\}$ with $\chi=(\zeta,w)$ and $T_{C\cap\overline U}(\chi)=T_{C}(\chi)$ as $\chi\in C\cap\interior(U)$. We also have that for any $\chi\in C\cap\partial U$ and almost all $t\in[0,\varepsilon)$, $\{0\}\subset T_{C\cap\overline{U}}(\chi)$ therefore $H_r(t,\chi)\cap T_{C\cap\overline U}(\chi)\neq \emptyset$. We can thus apply \cite[Thm.~2.3]{Carjua2000viability} to conclude that for any $\chi_0\in C\cap\overline U$, there exists a solution to
    \begin{equation}
        \begin{array}{rlll}
            \dot\chi\in H_r(t,\chi) & & \chi\in C\cap \overline U
        \end{array}\label{eq:sys-Hr}
    \end{equation}
    initialized at $\chi_0$ that is defined over $[0,\varepsilon_r)$ for some $\varepsilon_r\in(0,\varepsilon]$. Finally, the latter is now exploited for the original system
    \begin{equation}
        \begin{array}{rllll}
            \dot z\in F(z,\omega) & & (z,\omega)\in C.
        \end{array}\label{eq:sys-constrained-differential-inclusion}
    \end{equation}
    Given absolutely continuous $w$, consider $\xi\in C$ and a neighborhood $U$ of $\xi$. We know that there exists a solution $(\zeta,w)$ to (\ref{eq:sys-Hr}) that is defined over some interval of time $[0,\varepsilon_r)$ with $\varepsilon_r>0$. Since $(\xi,w(0))\in C\cap\interior(U)$ and $(\zeta,w)$ is continuous, there exists $\widehat{\varepsilon}_r\in(0,\overline\varepsilon_r)$ such that $(\zeta(t),w(t))\in C\cap\interior (U)$. As $H_r=H$ on $C\cap\interior(U)$, we derive that $(\zeta,w)$ is an e-solution to (\ref{eq:sys-constrained-differential-inclusion}), i.e. VC-e$(\xi,w)$ holds, which concludes the proof.
\end{proof}

\begin{remark}
    A challenge in the proof of Proposition \ref{prop:v:ac} is that most viability results in the literature, such as \cite{Carjua2000viability}, provide `global' conditions guaranteeing existence of solutions to constrained differential inclusions for {\em all} initial states in the constraint set. Since we are interested to derive `local' conditions for existence of solutions for a particular initial state (and given input function) extra care and arguments are needed to apply the global viability results to obtain local viability results. We will see that also in the proof of Proposition~\ref{prop:V:split} similar arguments have to be used. There are only few local viability results in the literature, such as \cite[Thm.~3.4.2]{Aubin2009} and \cite[Lem.~5.26]{Goe12}, but they both apply to time-independent (and thus input-free) scenarios only.
\end{remark}

\subsubsection{Continuous inputs}
We now relax the previous requirement on $w$ and assume that it is only continuous. This allows us to exploit \cite[Proposition 6.10]{Goe12} by adding the time as an extra state variable. Hereto, given continuous $w$, we consider the set-valued map $K_w: [0,\infty) \rightrightarrows \ree^{n_x}$ with
\begin{equation}
    K_w(t):=\{\zeta\in \ree^{n_x} \mid (\zeta,w(t))\in \C\} \quad \text{for} \quad t\in \ree_{\geq 0}.
\end{equation}
Note $\graph(K_w)= \{(t,\zeta)\in [0,\infty)\times \ree^{n_x} \mid (w(t),\zeta)\in C\}$.

\begin{proposition}\label{prop:continuous-w}
    Consider $\HS$ with Assumption~\ref{sa:reg} holding. Let $\xi\in \R^{n_x}$, $w$ be continuous on an interval $[0,\varepsilon)$ for some $\varepsilon>0$ and suppose that $K_w(t) \neq \emptyset$ for all $t\in [0,\varepsilon)$. If there exists a neighborhood $U$ of $\xi$ such that, for all $\zeta \in U$ and all ${\tau}\in[0,\varepsilon)$, satisfying $(\zeta, w(\tau)) \in C$,
    \begin{equation}
        \{1\} \times F(\zeta, w(\tau))\cap T_{\graph(K_w)}(\tau,\zeta) \neq \emptyset, \label{eq:vc-tc-C}
    \end{equation}
    then VC-e($\xi,w$) holds.
\end{proposition}

\begin{proof}
    We define a new state variable $\chi=(t,z)$ and embed $\dot z\in F(z,w(t))$ with $(z,w(t))\in C$ in the autonomous constrained differential inclusion
    \begin{equation}\label{eq:sys-proof-continuous-w-J}
        \dot\chi \in J(\chi) \quad \chi \in \graph(K_w)
    \end{equation}
    with $J(\chi) := \{1\}\times F(z,w(t))$. To apply \cite[Lem.~5.26]{Goe12}, we first observe that $\graph(K_w)$ is closed due to the continuity of $w$ and the closedness of $\C$. Moreover, $J$ is outer semicontinuous, locally bounded and takes non-empty convex values in view of its definition and Assumption~\ref{sa:reg}. Let $\xi\in\R^{n_x}$ and $w$ be as in the proposition. We now check the condition in \cite[Lemma 5.26(b)]{Goe12} for \eqref{eq:sys-proof-continuous-w-J} at initial state $(0,\xi)$. Consider a neighborhood $U$ of $\xi$ as in Proposition~\ref{prop:continuous-w}. Let $\nu\in(0,\varepsilon)$ and take $V:=U\times(-\nu,\nu)$ as a neighborhood of $(0,\xi)$. Let $(\tau,\zeta)\in \graph(K_w)\cap V$. Clearly, then $\tau\geq 0$ as $\graph(K_w)$ is only defined for $t\geq 0$. In case $\tau\geq 0$, (\ref{eq:vc-tc-C}) holds and thus for all $(\tau,\zeta)\in \graph(K_w)\cap V$ we have $J(\tau,\zeta)\cap T_{\graph(K_w)}(\tau,\zeta)\neq \emptyset$. Consequently, the condition in \cite[Lemma 5.26(b)]{Goe12} holds from which we derive that there exists a solution $\chi=(t,z)$ to (\ref{eq:sys-proof-continuous-w-J}) on $[0,\widehat\varepsilon)$ with $\widehat\varepsilon\in(0,\varepsilon]$. Hence, $\dot z\in F(z,w(t))$ for almost all $t\in[0,\widehat \varepsilon)$ and $(z(t),w(t))\in C$ for all $t\in[0,\widehat \varepsilon)$, i.e.,  VC-e$(\xi,w)$ holds.
\end{proof}

\subsection{Measurable inputs}

In this section, we allow inputs $w$ to be measurable, which becomes rather intricate when they affect the flow set $C$.

\subsubsection{Measurable inputs not affecting $C$}

We first consider the scenario in which the inputs $w$ can be split as $(w_1,w_2)$ with $w_1$ taking values in $\ree^{n_{w_1}}$ that affect $\C$ and are assumed to be AC, and $w_2$ with values in $\ree^{n_{w_2}}$ but not affecting $\C$, and the input $w_2$ is allowed to be measurable. Hence, we assume that $\C$ can be written as $\C = \C_1 \times \ree^{n_{w_2}}$.

\begin{proposition}\label{prop:V:split}
    Consider $\HS$ with Assumption~\ref{sa:reg} holding. Assume that $w=(w_1,w_2)$ with $\C = \C_1 \times \ree^{n_{w_2}}$, $\C_1 \subseteq\ree^{n_{x}}\times \ree^{n_{w_1}} $, $w_1$ AC with its derivative essentially bounded on an interval $[0,\varepsilon]$ for some $\varepsilon>0$, and $w_2$ is  measurable and essentially bounded on $[0,\varepsilon]$. If there exists a neighborhood $U$ of $(\xi,w_1(0))$ such that for all $(\zeta, \omega)$ with $(\zeta, \omega_1) \in \C_1 \cap U$ and almost all ${\tau}\in[0,\varepsilon)$
    \begin{equation}
        F(\zeta, \omega_1,w_2(\tau)) \times \{\dot w_1(\tau)\} \cap T_{C_1}(\zeta, \omega_1) \neq \emptyset, \label{eq:vc-tc-ACandM}
    \end{equation}
    then VC-e($\xi,w$) holds.
\end{proposition}
\begin{proof}
    The proof extends the proof of Proposition~\ref{prop:v:ac} by considering $\chi=(x,w_1)$ and embedding $\dot z\in F(z,w)$, $(z,w)\in \C$ in
    \begin{equation}\label{eq:H}
        \dot\chi \in \tilde{H}(t,\chi):= F(\chi,w_2(t)) \times \{v_1(t)\},\ \ \chi\in C_1
    \end{equation}
    where $v_1$ denotes $\dot{w}_1$. Hence, $w_1$ is embedded as a state variable but $w_2$ is not. Following similar steps, realizing that $\tilde H$ satisfies (A)-(D) in \cite{Carjua2000viability} taking the set $D$ in \cite[Thm.~2.3]{Carjua2000viability} as $D=\C_1 \cap \overline{U}$ as in the proof of Proposition~\ref{prop:v:ac}, using \cite{Carjua2000viability} establishes the result. Note that to show (A), i.e., $\tilde H$ satisfying the required measurability of $t\mapsto \tilde H(t,\chi)$ for fixed $\chi=(\zeta,\omega_1)\in C_1\cap U$, we use \cite[Thm.~14.13(b)]{rockafellar-wets-book} for fixed $\zeta,\omega_1$ in the following way: We take for $S$ in \cite[Thm.~14.13(b)]{rockafellar-wets-book} $S:=w_2$, which is closed-valued (it has single values) and is measurable on $[0,\varepsilon]$ by assumption and $M(t,\cdot): \omega_2\mapsto F(\zeta,\omega_1,\omega_2)$ is osc due to Assumption~\ref{sa:reg}, for given fixed $\zeta,\omega_1$. Moreover, note that $M(t,\cdot)$ is the same for all $t$, once $\zeta$ and $\omega_1$ are fixed). Hence, from \cite[Thm.~14.13(b)]{rockafellar-wets-book}, it follows that $t\mapsto F(\zeta,\omega_1,\omega_2(t))$ is measurable on $[0,\varepsilon]$ for each fixed $\chi=(\zeta,\omega_1)\in C_1\cap U$. Hence, as $v_1$ is measurable, also $\tilde H$ is measurable on $[0,\varepsilon]$ and thus (A) holds. (B)-(D) in \cite{Carjua2000viability} can be obtained as in the proof of Proposition~\ref{prop:v:ac}. Finally, the remainder of the proof follows as in the proof of Proposition~\ref{prop:v:ac} considering \eqref{eq:Hr} with $H$ replaced by $\tilde H$ as in \eqref{eq:H}.
\end{proof}

Note that Proposition~\ref{prop:V:split} has Proposition~\ref{prop:v:ac} as a special case by taking $w=w_1$ (and $w_2$ being absent) requiring $w$ to be locally AC. In addition, Proposition~\ref{prop:V:split} has the result of \cite{Carjua2000viability} in which $C$ is not dependent on $w$ (in the context of \cite{Carjua2000viability} not depending on time $t$) as a special case (take $w=w_2$ in Proposition~\ref{prop:V:split}). Finally, the results in \cite{frankowska1996measurable} are also closely related to Proposition~\ref{prop:V:split} as the constraint set is allowed to depend on time, but in an AC manner, similar to our assumption that $w_1$ affecting $C$ has to be (locally) AC.

\subsubsection{Measurable inputs affecting $\C$}
\label{subsubsec:measu}

In case the inputs affecting $\C$ are not continuous, and, thus $w$ can abruptly change at each time instant due to only being measurable, the sufficient conditions have to take into account that $\C$ ``seen from the perspective of the state $x$'' can change instantaneously, or in terms of \cite{frankowska1996measurable}, the constraint set does not depend (absolutely) continuously on time $t$. Stronger conditions are needed in this case, as stated in the next proposition.

\begin{proposition} \label{prop:CD}
    Consider $\HS$ with Assumption~\ref{sa:reg} holding. Let $\xi\in\ree^{n_x}$ and $w\in {\cal L}_{\W}$ be given, and assume that  there exists a $\delta>0$ such that $\mathbb{B}(\xi,\delta)\times {\W}\subseteq {\C}$. Then \vcaexi and \vcexi hold.
\end{proposition}
\begin{proof}
    Due to the standing assumptions, there exist $\varepsilon>0$ and an absolutely continuous function $z:[0,\varepsilon]\rightarrow\R^{n_x}$ such that $z(0)=\xi$, $\dot z(t) \in F(z(t),w(t))$ for almost all $t\in [0,\varepsilon]$, see, e.g, \cite{Carjua2000viability}, where we used again \cite[Thm.~14.13]{rockafellar-wets-book} to show that $t\mapsto F(z,w(t))$ for fixed $z$ satisfies the required measurability properties. Note that this solution $z$ is unrelated to $\C$; it just states that a local solution exists to the differential inclusion without the constraints. However, due to the hypothesis in the proposition, there exists a $\delta>0$ such that $(\mathbb{B}(\xi,\delta),w(t)) \subseteq {\C}$ for all $t\in \ree_{\geq 0}$. Since $z$ is continuous and $z(0)=\xi$, $(z(t),w(t)) \in {\C}$ for all $t\in [0,\delta']$ for some $0<\delta' < \delta$ and, hence, \vcaexi and \vcexi hold.
\end{proof}

The condition in Proposition \ref{prop:CD} is a geometric check that there is some ``margin'' around $\xi$ to the boundary of $C$ for all admissible values of the input. Clearly, if this holds for all $\xi$ in a set $\Xi_0\subseteq \R^{n_x}$, then \vcset holds, as needed in statement \emph{(ii)} in Proposition~\ref{prop:210}.

To show the use of Proposition~\ref{prop:CD}, consider $C$ and $D$ as
\begin{subequations} \label{eq:outputsets}
    \begin{eqnarray}
        {\C} & = &\{(\zeta,\omega) \in \ree^{n_x} \times \W \mid h(\zeta)+ \omega \in {\C}_y\}, \\ {\D} & = & \{(\zeta,\omega) \in \ree^{n_x} \times \W\mid h(\zeta)+ \omega \in {\D}_y\},
    \end{eqnarray}
\end{subequations}
where $\C_y\subseteq \ree^{n_y}$ and $\D_y\subseteq \ree^{n_y}$ and $h:\ree^{n_x}\rightarrow \ree^{n_y}$. This case is relevant in, e.g., event-triggered control settings with measurement noise \cite{Scheres_Postoyan_Heemels_2024,BorDol_TAC18a,Mousavi_et-al-cdc2019}, see also Example~\ref{ex:noise} above.

\begin{lemma} \label{lem:CD}
    Consider sets ${\C}$ and $\D$ as in \eqref{eq:outputsets} for a continuous function $h:\ree^{n_x} \rightarrow \ree^{n_y}$, closed sets $\C_y\subseteq \ree^{n_y}$ and $\D_y\subseteq \ree^{n_y}$. Let ${\W}$ be given. For all $\xi\in\C_0\cap \Pi_x(\D^c,\W)$ the condition of Proposition~\ref{prop:CD} holds, if
    \begin{equation} \label{eq:sets}
        \range(h) \cap (\C_y - \W) \cap ( \D_y^c-\W)\subseteq \interior( {\C_y}\ominus W),
    \end{equation}
    where $\range(h):=\{ h(x) \mid x\in \ree^{n_x}\}$. In case $h$ is also an open map in the sense that it maps open sets to open sets, this condition is also necessary.
\end{lemma}
\begin{proof}
    We have ${\C}_0= \{ \zeta \mid h(\zeta) \in \C_y - {\W} \}$ and $\Pi_x(\D^c,\W) = \{ \zeta \mid h(\zeta) \in \D_y^c - {\W} \}$. Hence,
    \[
        \C_0\cap \Pi_x(\D^c,\W) = \{ \zeta \mid h(\zeta) \in (\C_y - {\W}) \cap (\D_y^c - {\W}) \}.
    \]
    This shows that item \emph{(ii)} of Proposition~\ref{prop:CD} holds  for all $\xi\in \C_0\cap \Pi_x(\D^c,\W)$, if and only if all $\xi$ with \[ h(\xi) \in (\C_y - {\W}) \cap (\D_y^c - {\W})\] have a $\delta>0$ such that $h(\B(\xi,\delta))+ {\W}\subseteq { \C}_y$, i.e., in terms of the Pontryagin difference $h(\B(\xi,\delta)) \subseteq ( {\C}_y\ominus \W).$ When $h$ is continuous, this is guaranteed under the stated conditions, as $h(\B(\xi,\delta))$ is a subset of $\B(h(\xi),\varepsilon)\cap \range(h)$, where $\varepsilon>0$ can be made arbitrarily small by choosing $\delta>0$ small. When $h$ is an open map, the condition \eqref{eq:sets} is clearly necessary, as $h(\B(\xi,\delta))$ contains $\B(h(\xi),\varepsilon')$ for some $\varepsilon'>0$.
\end{proof}

\begin{example}
    We revisit Example~\ref{ex:noise} and demonstrate the application of Lemma~\ref{lem:CD}. Note that $h(x)=x$. Moreover, $C=[-1.5,1.5]$, $D^c=[-1,1]$, $\W:= [-0.2,0.2]$. We get $D^c-\W=[-1.2,1.2]$ and ${{C}}-\W=[-1.7,1.7]$ and thus $({D}^c-\W)\cap ({{C}} - \W) = [-1.2,1.2]$. Finally, $\interior( {C}\ominus \W) =\interior ([-1.5,1.5]\ominus [-0.2,0.2])=(-1.3,1.3)$. Thus, \eqref{eq:sets} indeed holds as $[-1.2,1.2]\subseteq (-1.3,1.3)$. Hence, since item (i) for Proposition~\ref{prop:CD} is trivially satisfied for Example~\ref{ex:noise} the existence of nontrivial solutions is guaranteed for all $\sol(0,0)=\xi \in {\C}_0$ and all $w\in{\cal L}_{\W}$ based on Proposition~\ref{prop:210}.
 \end{example}

\section{Conclusions} \label{sec:conclusions}
We presented solution concepts for hybrid dynamical inclusions affected by continuous-time input signals. Conditions for the existence and the completeness were given involving trajectory-dependent viability conditions (for both measurable and piecewise continuous input signals). Afterwards, we leveraged and developed viability results for non-autonomous constrained differential inclusions to arrive at trajectory-independent viability conditions using tangent cones.

\bibliographystyle{IEEEtran}
\bibliography{references}

@article{prieur-teel-tac11,
    title={Uniting local and global output feedback controllers},
    author={C. Prieur and A. R. Teel},
    journal={IEEE Trans. Aut. Contr.},
    volume={56},
    number={7},
    pages={1636--1649},
    year={2011}
}

@book{Goe12,
    title={{Hybrid Dynamical Systems: Modeling, Stability, and Robustness}},
    author={Goebel, R. and Sanfelice, R. G. and Teel, A. R.},
    isbn={9780691153896},
    year={2012},
    publisher={Princeton University Press}
}

@book{Aubin1984,
    title     = {Differential Inclusions: Set-Valued Maps and Viability Theory},
    author    = {Jean-Pierre Aubin and Arrigo Cellina},
    series    = {Grundlehren der mathematischen Wissenschaften},
    volume    = {264},
    year      = {1984},
    publisher = {Springer-Verlag},
    address   = {Berlin, Germany},
}

@book{Aubin2009,
    title     = {Viability Theory},
    author    = {Jean-Pierre Aubin},
    year      = {2009},
    publisher = {Birkhauser, Boston},
}

@book{liberzon2003switching,
    title={Switching in Systems and Control},
    author={D. Liberzon},
    year={2003},
    publisher={Springer}
}

@inproceedings{khalil2024hybrid,
    title={Hybrid low-dimensional limiting state of charge estimator for multi-cell lithium-ion batteries},
    author={M. Khalil and R. Postoyan and S. Ra{\"e}l and D. Nesi{\'c}},
    booktitle={IEEE  Conf.~Decision and Control},
    pages={7828--7833},
    year={2024},
}

@article{saez2025hybrid,
    title={Hybrid control of MISO systems with delays},
    author={J. F. S{\'a}ez and A. Ba{\~n}os and A. Arenas},
    journal={ISA Transactions},
    volume={161},
    pages={216--227},
    year={2025},
}

@article{carjua2000viability,
    title={Viability results for nonautonomous differential inclusions.},
    author={O. C{\^a}rj{\u{a}} and M. D. P. {Monteiro Marques}},
    journal={Journal of Convex Analysis},
    volume={7},
    number={2},
    pages={437--443},
    year={2000}
}

@article{frankowska1996measurable,
    title={A measurable upper semicontinuous viability theorem for tubes},
    author={H. Frankowska and S. Plaskacz},
    journal={Nonlinear Analysis: Theory, Methods \& Applications},
    volume={26},
    number={3},
    pages={565--582},
    year={1996},
}

@inproceedings{petri-et-al-cdc2022,
    title={Towards improving the estimation performance of a given nonlinear observer: a multi-observer approach},
    author={E. Petri and R. Postoyan and D. Astolfi and D. Ne{\v{s}}i{\'c}, D and V. Andrieu},
    booktitle={IEEE Conf.~Decision and Control},
    year={2022}
}

@article{Astrom1999comparison,
  author={K. J. {\AA}str\"{o}m and B. M. Bernhardsson},
  title={Comparison of periodic and event based sampling for first order stochastic systems},
  journal={In Proceedings of the 14th IFAC World Congress, Beijing, China},
  pages={301--306},
  year={1999}
}

@inproceedings{Mousavi_et-al-cdc2019,
  title={On integral input-to-state stability of event-triggered control systems},
  author={S. H. Mousavi and N. Noroozi and R. Geiselhart and M. K{\"o}gel and R. Findeisen},
  booktitle={IEEE  Conf. Decision and Control},
  pages={1674--1679},
  year={2019}
}

@INBOOK{HeeCam_BOOK_CONTRIB03a, 
    title = {On the existence and uniqueness of solution trajectories to hybrid dynamical systems},
    author = {W. P. M. H. Heemels and M. K. Camlibel and {Van der Schaft}, A. J. and J. M. Schumacher},
    year = {2002},
    pages = {391--422},
    series = {Nonlinear and Hybrid Systems in Automotive Control},
    publisher= {Springer-Verlag}
}

@article{LygJoh:2003,
    title = {Dynamical Properties of Hybrid Automata},
    author = {J. Lygeros and  K. H. Johansson and  S. N. Simic and J. Zhang and S. S. Sastry},
    volume = 48,
    number = 1,
    year = 2003,
    journal = {IEEE Trans. Aut. Control}
}

@article{Cortes, 
    journal = {IEEE Control Systems Magazine},
    volume = {28},
    number = 3,
    year = {2008},
    title = {Discontinuous dynamical systems},
    author = {J. Cortes}
}

@ARTICLE{BorDol_TAC18a, 
    AUTHOR = {D. P. Borgers and V. S. Dolk and W. P. M. H. Heemels},
    TITLE = {Riccati-Based Design of Event-Triggered Controllers for Linear Systems with Delays},
    JOURNAL = {IEEE Trans.~Aut.~Control},
    MONTH = {},
    YEAR = {2018},
    VOLUME = {63},
    ISSUE = {1},
    PAGES = {174--188},
}

@article{NesTee_Aut_04,
    title = {Input-to-state stability of networked control systems},
    year = 2004,
    journal = {Automatica},
    volume = 40,
    pages = {2121--2128},
    author = {Ne\v{s}i\'{c}, D. and Teel, A. R.}
}

@ARTICLE{HeeTee_TAC10a, 
    AUTHOR = {W. P. M. H. Heemels and A. R. Teel and N. van de Wouw and D. Ne\v{s}i\'{c}},
    TITLE = {Networked Control Systems with Communication Constraints: Tradeoffs between Transmission Intervals, Delays and Performance},
    JOURNAL = {IEEE Trans. Aut. Contr.},
    YEAR = {2010},
    VOLUME = {55},
    ISSUE = {8},
    PAGES = {1781--1796},
}

@ARTICLE{pri_tar_zac_AUT13,
    author = {Prieur,C. and Tarbouriech,S. and Zaccarian,L.},
    title = {Lyapunov-based hybrid loops for stability and performance of continuous-time control systems},
    journal = {Automatica},
    year = {2013},
    volume = {49},
    pages = {577--584},
    number = {2}
}

@ARTICLE{nes_zac_tee_AUT08,
    author = {Ne\v{s}i\'{c},D. and Zaccarian,L. and Teel,A. R.},
    title = {Stability properties of reset systems},
    journal = {Automatica},
    year = {2008},
    volume = {44},
    pages = {2019--2026},
    number = {8},
}

@ARTICLE{LooGru_AUT17a, 
    AUTHOR = {S. J. L. M. van Loon and K. Gruntjens and M. F. Heertjes and N. van de Wouw and W. P. M. H. Heemels },
    TITLE = {Frequency-domain tools for stability analysis of reset control systems},
    JOURNAL = {Automatica},
    MONTH = {},
    YEAR = {2017},
    VOLUME = {82},
    ISSUE = {},
    PAGES = {101-108},
}

@article{Cai2009,
    author = {Cai, Chaohong and Teel, Andrew R.},
    doi = {10.1016/j.sysconle.2008.07.009},
    issn = {01676911},
    journal = {Syst. {\&} Contr. Lett.},
    number = {1},
    pages = {47--53},
    publisher = {Elsevier B.V.},
    title = {{Characterizations of input-to-state stability for hybrid systems}},
    volume = {58},
    year = {2009}
}

@article{zhuang2024robust,
    title={Robust multi-rate fusion state estimation for networked nonlinear systems via a dynamic event-timing-triggered mechanism},
    author={X. Zhuang and Y. Tian and H. Wang},
    journal={Information Sciences},
    volume={684},
    pages={121295},
    year={2024},
    publisher={Elsevier}
}

@book{goebel-book2024,
    title={{Set-Valued, Convex, and Nonsmooth Analysis in Dynamics and Control: An Introduction}},
    author={R. K. Goebel},
    year={2024},
    publisher={SIAM}
}

@article{Carjua2000,
    title={Viability results for nonautonomous differential inclusions.},
    author={C{\^a}rj{\u{a}}, Ovidiu and Monteiro Marques, Manuel DP},
    journal={Journal of Convex Analysis},
    volume={7},
    number={2},
    pages={437--443},
    year={2000},
    publisher={Heldermann Verlag}
}

@book{rockafellar-wets-book,
    title={Variational Analysis},
    author={R.T. Rockafellar and R.J.-B. Wets},
    year={2009},
    publisher={Springer}
}

@article{Bernard2020,
    title={{Hybrid dynamical systems with hybrid inputs: Definition of solutions and applications to interconnections}},
    volume={30},
    ISSN={1049-8923, 1099-1239},
    number={15},
    journal={Int. J. Robust and Nonlinear Control},
    author={Bernard, Pauline and Sanfelice, Ricardo G.},
    year={2020},
    pages={5892--5916}
}

@article{Chai2019,
    title={Forward Invariance of Sets for Hybrid Dynamical Systems (Part I)},
    volume={64},
    ISSN={0018-9286, 1558-2523, 2334-3303},
    number=6,
    journal={IEEE Transactions on Automatic Control},
    author={Chai, Jun and Sanfelice, Ricardo G.},
    year={2019},
    month={Jun},
    pages={2426--2441}
}

@Book{220,
    title        = {Hybrid Feedback Control},
    publisher    = {Princeton Univ.~Press},
    year         = {2021},
    author       = {R. G. Sanfelice}
}

@article{Scheres_Postoyan_Heemels_2024,
    title={Robustifying event-triggered control to measurement noise},
    volume={159},
    ISSN={00051098},
    DOI={10.1016/j.automatica.2023.111305},
    journal={Automatica},
    author={Scheres, K.~J.~A. and Postoyan, Romain and Heemels, W.~P.~M.~H.},
    year={2024},
    pages={111305}
}

\end{document}